\documentclass[conference]{article}

\makeatletter

\usepackage{blindtext}
\usepackage{eso-pic}
\usepackage{cite}
\usepackage{amsmath,amssymb,amsfonts}
\usepackage{amsthm}
\usepackage{mathabx}
\usepackage{algorithm}
\usepackage{algorithmic}
\usepackage{graphicx}
\usepackage{textcomp}
\usepackage{xcolor}

\def\BibTeX{{\rm B\kern-.05em{\sc i\kern-.025em b}\kern-.08em
    T\kern-.1667em\lower.7ex\hbox{E}\kern-.125emX}}
    
\usepackage{eso-pic}

\usepackage{indentfirst}

\usepackage{fullpage}
\usepackage{amsthm}
\usepackage{multirow}

\usepackage{pifont}
\usepackage{comment}

\theoremstyle{definition}

\newtheorem{remark}{Remark}
\newtheorem{theorem}{\textbf{Theorem}}

\newcommand{\B}{\mathcal B}

\date{May 11, 2025}

\begin{document}

\title{\vspace*{1cm} Residue Number System Comparison revisited, a software perspective
%{\footnotesize \textsuperscript{*}Note: Sub-titles are not captured in Xplore and should not be used}
\thanks{This work has been partially funded by the AID, project 2022151.}%Identify applicable funding agency here. If none, delete this.}
}
%\author{\IEEEauthorblockN{ Laurent-Stéphane Didier, Léa Glandus, Jean-Marc Robert}
%\IEEEauthorblockA{\textit{Laboratoire IMath} \\
%\and
%\IEEEauthorblockN{Nadia El Mrabet}
%\IEEEauthorblockA{\textit{Centre CMP, Department SAS} \\ %TODO check
%\textit{Mines Saint-Etienne,CEA-LETI}\\
%Saint-Etienne, France \\
%nadia.el-mrabet@emse.fr}
%\and
%\IEEEauthorblockN{\textsuperscript{}}
%\IEEEauthorblockA{\textit{} \\
%\textit{}\\
%\\
%}
%\and
%\IEEEauthorblockN{3\textsuperscript{rd} Léa Glandus}
%\IEEEauthorblockA{\textit{Laboratoire IMath} \\
%\textit{Université de Toulon}\\
%Toulon, France \\
%lea.glandus@univ-tln.fr}
%\and
%\IEEEauthorblockN{4\textsuperscript{th} Jean-Marc Robert}
%\IEEEauthorblockA{\textit{Laboratoire IMath} \\
%\textit{Université de Toulon}\\
%Toulon, France \\
%jean-marc.robert@univ-tln.fr}

%}

\author{Laurent-Stéphane Didier\\
\textit{Université de Toulon}\\
Toulon, France \\
laurent-stephane.didier@univ-tln.fr
\and
Nadia El Mrabet\\
Centre CMP, Department SAS \\ %TODO check
\textit{Mines Saint-Etienne,CEA-LETI}\\
Saint-Etienne, France \\
nadia.el-mrabet@emse.fr
\and
Léa Glandus\\
\textit{Université de Toulon}\\
Toulon, France \\
lea.glandus@univ-tln.fr
\and
Jean-Marc Robert\\
\textit{Université de Toulon}\\
Toulon, France \\
jean-marc.robert@univ-tln.fr}

\maketitle
%\conf{\textit{  V. International Conference on Electrical, Computer and Energy Technologies (ICECET 2025) \\ 
%	3-6 July 2025, Paris-France}}
\begin{abstract}
This paper presents a novel method to compare two numbers in Residue Number System (RNS) using an additional modulus, which is often already available because it is required in modular computations and digital signal processing scaling.

Our method provides the comparison of two integers in the full range of the RNS base. It does not require moduli of a special form, unlike other state-of-the-art methods that are restricted to specific RNS bases or require bounds on input numbers.
Our approach only requires one single conversion to a mixed radix representation with a complexity of $\mathcal O(n^2)$, which can be reduced to $\mathcal{O}(\log(n))$ in time with parallelization. This provides a significant advantage over classical methods and more recent competitive methods which work under restrictions.
This opens perspectives for advancements in challenging RNS operations such as division, scaling, and cryptographic applications.

%We also present software implementations making use of SIMD instruction set. In case of n+1 modulus base, these operations present O(n) complexity in time (O(n\^\ 2) in operations). This leads to RNS division in O(n)log(n) in time and O(n\^\ 2)log(n).

%This document is a model and instructions for \LaTeX.
%This and the IEEEtran.cls file define the components of your paper [title, text, heads, etc.]. *CRITICAL: Do Not Use Symbols, Special Characters, Footnotes, or Math in Paper Title or Abstract.
\end{abstract}

%\copyrightnotice{XXX-X-XXXX-XXXX-X/XX/\$XX.00 ©20XX IEEE}

%\begin{keywords}
%Residue number system, comparison, Mixed-Radix, base-extension
%\end{keywords}

\section{Introduction}

This paper deals with comparison in Residue Number Systems (RNS), which is a costly operation that usually requires conversions to a positional system suitable to perform the comparison itself. 

RNS is a non-positional number system, that was described in 1959 by Garner \cite{garn59},  where operations such as additions, subtractions, and multiplications can be performed in parallel on independent residues (see Knuth \cite{knuth2014art}). A $n$-channel RNS base consists in a set of $n$ relatively prime residues. The dynamic range of such systems is the product of the moduli.

Such number systems mainly target signal processing and cryptographic applications (see Molahosseini \textit{et al.}\cite{molahosseini2017embedded}).
However, because of their non-positional nature, the comparison of two RNS numbers is costly. Furthermore, the division operation also remains costly, as it requires one to know the magnitude of both operands.

Several contributions addressed the comparison problem. In general, comparing numbers in RNS requires a conversion to a positional system such as Mixed Radix System (MRS) or binary. After this conversion of each number, the comparison is done by comparing each pair of digits in the lexicographic order. Almost all previously published techniques either are based on Mixed Radix conversion either make use of the Chinese Remainder Theorem (CRT) or the Core function (see Ananda) \cite{MOHAN2017}. 

The advantage of MRS conversion is that it involves computation on data that have the size of RNS moduli (see Szabo and Tanaka\cite{szabo-tanaka67}). However, the computation of the MRS digits has a quadratic complexity in operation, though it can be implemented in a parallel to execute in linear time.

The techniques based on the CRT formula require costly modular operations on large numbers having the magnitude of the RNS base. The use of specific bases can significantly reduce this cost in hardware implementations (see Zhang \textit{et al.}, Parhami, Wang \textit{et al.} and Mohan.\cite{zsy93, parhami93, wang1999new, mohan2016residue}), but are not suitable for large RNS bases that are used in cryptographic applications. 

In the implementation of the RNS division, some adaptations of the CRT formula have been made in order to obtain the magnitude of the operands. For their division, Lu and Chiang presented in \cite{lu1992novel} a method based on parity checking, which uses lookup tables. Unfortunately, this method is not practicable for large RNS bases. Similarly, Sousa proposed a comparison scheme also based on parity checking which is simplified by the use of a specific RNS base \cite{sousa2007efficient}.

Approximate techniques for computing the CRT have been proposed in division methods (see Posch \textit{et al.} and Bajar \textit{et al.} \cite{posch1996division, Bajard1998}). However, because of the rounding, they are not adapted to catch the magnitude of small values. Kawamura \textit{et al.} proposed a comparison method suitable for computing the final step of the modular multiplication of large RNS numbers \cite{Kawamura2000}. Recently, a comparison method based on the approximate computing of the CRT formula has been published (see Xiao \textit{et al.} \cite{xiao2016algorithms}). However, similarly to the methods of \cite{posch1996division, Bajard1998}, it allows only the comparison of numbers that are large enough.

An other way to compare numbers is to use  mapping functions from the RNS to binary that are different from the CRT formula. The Core functions proposed by Akushkii \textit{et al.} \cite{akushskii1977new} are such mapping functions. However, they require modular computations on very large numbers and some numbers cannot be used for comparison (see Miller \textit{et al.}\cite{miller1986analysis}). 
Dimauro \textit{et al.} developed a comparison scheme based on another mapping function called Diagonal function \cite{dimauro1993new}. The modular computations make this approach not suitable for cryptographic-size RNS numbers.

\paragraph*{Contributions}

We target software implementations of cryptographic applications which require computations with large numbers of size ranging around several thousand bits. In such a context, the most frequent operation is the modular multiplication that requires to deal with two different RNS bases (see Bajard \textit{et al.}\cite{Bajard2001}).

%{\color{red} Software aspect, based on base extension, cryptographic size, modular multiplication that requires two bases, crypto homomorphique}

In this paper, we revisit the comparison for large RNS numbers and present an approach which requires only one conversion instead of two in the classical methods of Szabo and Tanaka in \cite{szabo-tanaka67}. This is possible when we already have the representation of both numbers in both bases, and one of the modulus of the second base can be used as a redundant modulus for comparison purpose. Thus, we assume that the redundant modulus is readily available.

The comparison algorithm we propose works for all couple of numbers in the range of the RNS base, denoted $M$:

$$\forall N_1,N_2 \mbox{~such that~} 0\leq N_1,N_2<M$$

Our algorithm returns:

$$\mbox{True if~}N_1\geq N_2$$
or
$$\mbox{False if~}N_1<N_2$$

It requires a redundant modulus and uses a mixed radix extension to this modulus in order to provide the comparison. Furthermore, our approach is generic in the sense that there is no condition on the RNS base: no special form of the moduli, no restrictions in the modulus number, no bounds on the input values $N_1$ and $N_2$ insofar as they remain in the RNS dynamic range (i.e.  $0\leq N_1,N_2 <M$).

We provide a theorem with the proof of correctness of our algorithm and finally, a complexity evaluation which shows the competitiveness of our approach in our context.

We point out that our approach is designed to be, as far as possible, the most generalist one, and works whatever the size of the considered dynamic range of the RNS system.
%TODO si on veut
\paragraph*{Organization of the paper}
This paper is organized as follows : Section \ref{sec:RNS} reminds the RNS system and details base extension approaches, Section \ref{sec:compapp} presents our comparison method and its complexity, and a conclusion ends the paper.
 
\section{Residue Number Systems}
\label{sec:RNS}
In RNS, a number X is represented by its remainders $x_i= \left| X \right|_{m_i} = X \bmod m_i$. If the residues are relatively prime, the CRT proves that X is unique in [0, M[. We note $M=\prod_{i=1}^n m_i$.  The proof of this theorem gives a formula in order to compute $X$ with the remainders $x_i$:
\begin{equation}\label{eq_CRT}
    X = \left| \sum_{i=1}^n \left| x_i\times M^{-1}_i \right|_{m_i} \times M_i \right|_M
\end{equation}
where $M_i=M/m_i$ and $ \left| M^{-1}_i \times M_i\right|_{m_i} =1$. The base of this Residue Number System is the set of the $n$ relatively prime moduli $ \B=\{m_1, \cdots, m_n\} $.

%\medskip
The main advantage of this number representation system is the parallel computations of additions, subtractions and multiplications on each channel independently. Considering $X$ and $Y$, one computes in RNS $Z\gets X\odot Y$ as follows:

\[
	z_i = x_i \odot y_i \bmod m_i \mbox{ where } \odot \in\{+, -, \times\}
\]

However, the non-positional characteristic of the RNS representation renders rather difficult comparison, division or scaling operations.

This motivates the numerous works around comparison algorithms, which is the subject of the present work.

%TODO à finir, présentation de sopérations, et explications de la difficultés des opérations telles que division, comparaison, scaling...

\subsection{Base extensions}
\label{ssec:BE}
The base extension consists in converting an RNS number from an RNS base to another. This operation requires to convert the RNS number into a positional representation and to compute the residues of this value in the new RNS base.  There are mainly two families of base extension:

\begin{itemize}
    \item using a forward conversion to a Mixed Radix representation (MRC) and a backward computation of the value of this representation modulo each modulus of the destination base. 
    \item approaches based on the Chinese Remainders Theorem (CRT) formula.
\end{itemize}

\subsubsection{Mixed Radix technique}

This method was presented by Szabo and Tanaka in \cite{szabo-tanaka67}. It consists in converting the RNS representation into a positional number system using the partial products of moduli as follows:

\begin{equation}
    \label{eq:MRConv}
    \begin{array}{lcl}X &=& a_1 + a_2\cdot m_1 + a_3\cdot m_1 \cdot m_2 + \hdots\\
			 && + a_{n-1} \cdot m_1 \cdot \hdots\cdot m_{n-2}\\ &&+ a_n \cdot m_1 \cdot \hdots \cdot m_{n-1}
	\end{array}
\end{equation}

We denote $m_{j,i}^{-1}$ the inverse of $m_j$ $\bmod~m_i$, that is $m_j\cdot m_{j,i}^{-1}\equiv 1 \bmod m_i$.

The mixed-radix digits $a_i$ are computed as follows:

$$\left\{\begin{array}{lcl}a_1 &=& x_1 \bmod m_1\\
                    a_2 &=& (x_2-a_1)\cdot m_{1,2}^{-1} \bmod m_2\\
                    a_3 &=& ((x_3-a_1)\cdot m_{1,3}^{-1}-a_2)\cdot m_{2,3}^{-1} \bmod m_3\\
                    \vdots\\
                    a_n &=& (\cdots(x_n-a_1)\cdot m_{1,n}^{-1}-a_2)\cdot m_{2,n}^{-1})-\cdots\\&&-a_{n-1})\cdot m_{n-1,n}^{-1} \bmod m_n
                    
	\end{array}\right.$$

The main advantage of this technique is that there is no need for large integer computations and that the value of $X$ is obtained without correction to deal with the dynamic range of the RNS system.

Another characteristic is that each digit is bounded, i.e. $0\leq a_i<m_j$. This guarantees that one has for sure $0\leq X<M$.

In terms of complexity, this approach requires $n(n-1)/2$ word size multiplications and memorized constants. However, there are some parallelization techniques that allow $\mathcal O(\log(n))$ execution time in some cases (see, for example, Huang in \cite{Huang83}).

\subsubsection{CRT based approaches}

This class of techniques retrieves $X$ with the equation (\ref{eq_CRT}). The main drawback is that we need to compute $X$ as a sum modulo $M$. This operation is very costly because $M$ is a large number.

%as follows:
%\begin{eqnarray} X= \sum_{i=1}^{n} |x_i\cdot y_i|_{m_i}\cdot M_i \mod M,\end{eqnarray}

%where $M = \prod_{i=1}^{n} m_i$,  $M_i = M/m_i$ et $y_i = M_i^{-1} \mod m_i$, pour $1 \leq i \leq n$.

Thus, by denoting $Y = \sum_{i=1}^n \left| x_i\times M^{-1}_i \right|_{m_i} \times M_i$, this is equivalent to find $k$ such that
%\sum_{i=1}^{n}x_i\cdot M_i\cdot y_i
\begin{equation}\label{eq:CRT}
    X=Y-k\cdot M
\end{equation}

Different methods have been proposed to compute $k$:

\begin{itemize}
    \item Posch and Posch in \cite{PoschP95} suggest performing a floating-point division of approximate values ($\pm 1$). Thus, this method can not be used when the exact value of $k$ is needed.% However, this method provides an approximate $k$, and does not fit in our case since we need an exact value.
    \item Kawamura \emph{et al.} in \cite{Kawamura2000} propose another method also computing an approximate value of $k$, trading the division of the previous approach by shiftings and multiplications by small constants. Like the previous method, there is a bound on $X$ beyond which the estimation of $k$ is not correct, by $1$. This bound depends on a special parameter $\alpha$ introduced by the authors ($0\leq \alpha <1$), and can be expressed as $(1-\alpha)\cdot M$, and this parameter is set according to the base extension. On average, the authors suggest taking $\alpha = 0.5$, which is feasible in case of use in successive modular reduction, but not in case when the exact reduction $\bmod M$ is required.% This method does not fit in our case again, since no guarantee is given that the obtained value is the exact one. %{\color{red} phrase précédente inadaptée...}
    \item Shenoy and Kumaresan in \cite{shenoyScaling1989,ShenoyK89} suggest using a redundant modulus $m_r$ in order to compute $k$ as follows:

    $$k = \left|\left(\sum_{i=1}^{n} \left||x_i\cdot y_i|_{m_i}\cdot M_i\right|_{m_r}\right) -x_r\right|_{m_r}$$

    provided $m_r$ is large enough.

%{\color{red} phrase suivante inadaptée}

    One issue with this approach is that it is necessary to guarantee 

    $$x_r = X \bmod m_r$$

    in any case. Otherwise, one may meet a situation in which the evaluation of $k$ is erroneous.
    
    %In our case, we need to determine the real value of $x_r$. 
    Indeed, if we need to extend the difference of two numbers $N_1$ and $N_2$ modulo $M$, with $0\leq N_1,N_2<M$, and if $N_1<N_2$, the base extension requires $x_r = |M +N_1-N_2|_{m_r}$. Because we will use $x_r = ||N_1|_{m_r}-|N_2|_{m_r}|_{m_r} \neq |M +N_1-N_2|_{m_r}$, the estimation of $k$ will not be correct.
    
\end{itemize}

%%%%%%%%%%%%%%%%%%%%%%%%%%%%%%%%%%%%%%%%%%%%%%%%%%%%%%
\section{Comparison Approach}
\label{sec:compapp}

%A similar idea has briefly been brought up by Hitz and Kaltofen in \cite{HitzK1995} but they considered using full base extensions instead of just using one more modulus.% They then suggest to perform a comparison by comparing the coefficients

In this section, we present our method for number comparison in RNS. 
To determine whether $N_1 \geq N_2$ or $N_1 < N_2$, we will use an extra modulus that can be any value, large or small, provided that it is prime with $M$. We denote this modulus $m_a$. One extra modulus has also been used by Lu et Chiang in \cite{lu1992novel} with $m_a=2$. However, their design targets only hardware implementation and requires large look-up tables.

Here, we propose a comparison algorithm designed for software implementation of large RNS bases. This targets cryptographic applications.

\subsection{Algorithm}

In the context of cryptography, each operand of, say, modular multiplications, is represented in two different RNS bases $\B$ and $\B'$ \cite{Bajard2001}. This extra modulus $m_a$ can be chosen as one of the moduli of the second base $\B'$, provided that it is prime with each modulus of the first base $\B$. This is convenient because, as we will show, our method will only need one single conversion. In other contexts, this extra modulus can be managed separately.

This leads to the comparison Algorithm \ref{alg:RNScomp}. It takes as inputs two integers $N_1$ and $N_2$ in their RNS representation, respectively $\{n^{(1)}_1,\hdots,n^{(1)}_n\}$ and $\{n^{(2)}_1,\hdots,n^{(2)}_n\}$, and their residue modulo the redundant modulus, that is $n^{(1)}_a = N_1 \bmod m_a$ and $n^{(2)}_a = N_2 \bmod m_a$. The algorithm returns either $N_1\geq N_2$, or $N_1<N_2$. Algorithm \ref{alg:RNScomp} uses a conversion to a Mixed Radix representation, whose choice will be justified section \ref{ssec:BEChoice}. The proof of this algorithm is provided by the Theorem \ref{th:comparison}.

\begin{algorithm}
      \caption{Comparison of two integers in RNS representation, $RNSComp(N_1,N_2)$}
  \label{alg:RNScomp}
  \begin{algorithmic}[1]
    \REQUIRE $\{n^{(1)}_1,\hdots,n^{(1)}_n\}$ and $\{n^{(2)}_1,\hdots,n^{(2)}_n\}$, the RNS representation in base $\B$ of respectively $N_1$ and  $N_2$, $0\leq N_1,N_2<M$, and their residues $n^{(1)}_a = N_1 \bmod m_a$ and $n^{(2)}_a = N_2 \bmod m_a$
    \ENSURE $N_1\geq N_2$, or $N_1<N_2$
    \STATE $\Delta' \gets |n^{(1)}_a-n^{(2)}_a|_{m_a}$
    \STATE $\{z_1,\cdots,z_n\} \gets \{|n^{(1)}_1-n^{(2)}_1|_{m_1},\hdots,|n^{(1)}_n-n^{(2)}_n|_{m_n}\}$\\
    // Mixed Radix conversion, Alg. \ref{alg:MRConv}
    \STATE $\{\delta_1,\hdots,\delta_n\} \gets MRConversion(\{z_1,\cdots,z_n\})$\\
    // Residue $\bmod~m_a$, Alg. \ref{alg:m_aConv}
    \STATE $\Delta \gets to_{m_a}(\{\delta_1,\hdots,\delta_n\})$// $\Delta = |~|N_1-N_2|_M~|_{m_a}$
    \IF{$\Delta =\Delta'$}
        \RETURN $N_1\geq N_2$
    \ELSE
        \RETURN $N_1<N_2$
    \ENDIF
  \end{algorithmic}
\end{algorithm}

\begin{theorem}[Comparison]\label{th:comparison}
Let two RNS numbers $N_1$ and $N_2$ in the base $\B=\{m_1, \cdots, m_n\}$ and a modulus $m_a$ prime with $\B$. If $$N_1-N_2\bmod m_a =((N_1-N_2)\bmod M) \bmod m_a,$$ then $N_1 \geq N_2$, otherwise $N_1 < N_2$.
\end{theorem}
\begin{proof}

%First, we assume that $X-Y \notequiv 0 \bmod m_a $. In this case :

In one hand, we compute $\Delta_{\B} \gets N_1-N_2$ in base $\B$, i.e. $\Delta_{\B} = (N_1-N_2) \bmod M$. From this value, we compute $\Delta_{\B} \bmod m_a$. It can be computed as a small base extension of only 1 modulus $m_a$. Thus:
	$$\Delta = ((N_1-N_2)\bmod M) \bmod m_a$$
In the other hand,  we use $n^{(1)}_{m_a} \gets N_1\bmod m_a$ and $n^{(2)}_{m_a}\gets N_2 \bmod m_a$, in order to obtain:
$$\Delta' \gets (n^{(1)}_{m_a}-n^{(2)}_{m_a}) \bmod m_a$$
\begin{itemize}
	
	\item If $N_1 \geq N_2$, then:
    
   % $\Delta = \Delta'$.
	
	%Indeed, in such case :
	
	$$\begin{array}{rcl}\Delta &=& ((N_1-N_2)\bmod M) \bmod m_a\\
							&=& (N_1-N_2) \bmod m_a\\
							&=& (n^{(1)}_{m_a}-n^{(2)}_{m_a}) \bmod m_a\\
							&=& \Delta'
	\end{array}$$
	
	\item Otherwise, if $N_1<N_2$ then $N_1-N_2 <0$ and $(N_1-N_2)\bmod M = N_1-N_2+M$. Therefore, we have :
	
	$$\begin{array}{rcl}\Delta &=& ((N_1-N_2)\bmod M) \bmod m_a\\
							&=& (M+ N_1- N_2) \bmod m_a\\
							&=& ((M\bmod m_a) + n^{(1)}_{m_a}-n^{(2)}_{m_a}) \bmod m_a\\
							&\neq& \Delta'
	\end{array}$$
\end{itemize}

Reciprocally, if $\Delta \neq \Delta'$, then:
$$((N_1-N_2)\bmod M) \bmod m_a \neq (N_1 - N_2) \bmod m_a$$
and there exists an integer $k \neq 0$ so that 
$$((N_1-N_2)\bmod M) = N_1-N_2 +k.M$$
with $k=1$, because $0< (N_1-N_2)\bmod M < M$. Thus:
$$0<N_1-N_2 +M<M$$
and
$$N_1-N_2<0$$
Therefore, $N_1 < N_2$, otherwise $N_1 \geq N_2$.

%there exists $N_1 \geq N_2$, then 
%The reciprocal can be proved with boolean arguments. Let's take  {\color{red} L'argument ne me convaint pas vraiment}
%$$
%\left\{
%    \begin{array}{ll}
%        A : N_1 \geq N_2\\
%        B : \Delta = \Delta'\\
%    \end{array}
%\right.
%\implies
%\left\{
%    \begin{array}{ll}
%        \overline{A} : N_1<N_2\\
%        \overline{B} : \Delta \neq \Delta'\\
%    \end{array}
%\right.
%$$
%
%
%We have 
%$$
%\left\{
%    \begin{array}{ll}
%        A \implies B\\
%        \overline{A} \implies \overline{B}
%    \end{array}
%\right.
%$$

%So
%$$
%A \bigwedge \overline{B} \text{ impossible } \implies
%\left\{
%    \begin{array}{ll}
%        B \implies A\\
%        \overline{B} \implies \overline{A}
%    \end{array}
%\right.
%$$

%Hence,
%$$
%    \begin{array}{ll}
%        \Delta = \Delta' \Leftrightarrow N_1 \geq N_2\\
%        \Delta \neq \Delta' \Leftrightarrow N_1<N_2\\
%    \end{array}
%$$
\end{proof}

%{\color{red}Algorithm}
%We now explain our method to compare $N_1$ and $N_2$ both represented in RNS base $\B$ by $\{n^{(1)}_1,\hdots,n^{(1)}_n\}$ and $\{n^{(2)}_1,\cdots,n^{(2)}_n\}$ respectively. We know 

%\begin{itemize}

%	\item We compute $\Delta \gets N_1-N_2$ in base $\B$, i.e. $\Delta = (N_1-N_2) \bmod M$
    
%	\item We compute $\Delta \bmod m_a$ (as a small base extension of only 1 modulus $m_a$). So we get
%	$$\Delta = ((N_1-N_2)\bmod M) \bmod m_a$$

%	\item We use $n^{(1)}_{m_a} \gets N_1\bmod m_a$ and $n^{(2)}_{m_a}\gets N_2 \bmod m_a$, using the same process
	
%	\item We can then compute :
	
%	$$\Delta' \gets (n^{(1)}_{m_a}-n^{(2)}_{m_a}) \bmod m_a$$
	
%	\item If $N_1 \geq N_2$, then $\Delta = \Delta'$.
%    \item 
%\end{itemize}

\begin{remark}
    
Our method also works in the special cases where $N_1-N_2 \equiv 0 \bmod m_a $.

Indeed, in this case, one has $N_1 \bmod m_a \equiv N_2 \bmod m_a $, i.e. $n^{(1)}_{m_a} = n^{(2)}_{m_a}$ and one has always 

$$\Delta' \gets (n^{(1)}_{m_a}-n^{(2)}_{m_a}) \bmod m_a = 0$$

\begin{enumerate}
    \item if $N_1\geq N_2$, one has $N_1-N_2 = k\times m_a$, with $k$ a positive constant. Thus $\Delta = 0 = \Delta'$
    \item if $N_1 < N_2$, one has $N_1-N_2 = - k\times m_a \equiv M - k\times m_a \bmod M$, with $k$ a positive constant.
    Thus $\Delta = M \bmod m_a \neq \Delta'$, since $\Delta' = 0$.
\end{enumerate}
\end{remark}

\subsection{Choice of base extension}
\label{ssec:BEChoice}

Our comparison method works, provided we get the correct value $\bmod~m_a$. That is 

$$\forall X, 0 \leq X \leq M \\$$
and
$$ \{x_1,\hdots, x_n\}$$
the set of residues in the RNS base $\B = \{m_1,\hdots, m_n\}$, we need to compute $x_{m_a} = X \bmod m_a$

To achieve this, we use a base extension from the base $\B$ into the modulus $m_a$.

%\medskip

%None of these methods can be used for our purpose.

\subsubsection*{Choice of the method and algorithm}

Taking into account what has been said in Section \ref{ssec:BE}, and since it is mandatory to have the exact value $\bmod M$ (i.e. the exact value of $k$ in the computation of $N_1-N_2 \bmod M$, see equation \ref{eq:CRT}), the chosen base extension method is based on the \emph{Mixed Radix} conversion. We now present the generic algorithms for our comparison approach.

\paragraph{Mixed Radix conversion}

We provide the Mixed Radix conversion method Algorithm \ref{alg:MRConv}, as presented by Szabo and Tanaka\cite{szabo-tanaka67}. This algorithm provides the mixed radix representation of the number $X$ from its residues in the RNS base $\B$.

\begin{algorithm}
      \caption{Mixed Radix conversion, $MRConversion(\{x_1,\cdots,x_n\})$}
  \label{alg:MRConv}
  \begin{algorithmic}[1]
    \REQUIRE  $\{x_1,\hdots,x_n\}$  in RNS base $\B$ representing $X$, $0\leq X<M$, precomputed values $m_{j,i}^{-1} \bmod m_i$
    \ENSURE $\{a_1,\cdots,a_n\}$ as defined eq. \ref{eq:MRConv}
    \FOR{$i$ from $1$ to $n$}
        \STATE $a_i \gets x_i$
    \ENDFOR
    \FOR{$i$ from $2$ to $n$}
        \FOR{$j$ from $1$ to $i-1$}
            \STATE $a_i \gets (a_i-a_j)\cdot m_{j,i}^{-1} \bmod m_i$
        \ENDFOR
    \ENDFOR
    
    \RETURN $\{a_1,\cdots,a_n\}$
   
  \end{algorithmic}

\end{algorithm}

The complexity of this approach is straightforward: one performs $n\cdot(n-1)/2$ modular multiplications. In terms of memory, the need is $n\cdot(n-1)/2$ word constants, corresponding to the values $m_{j,i}^{-1} \bmod m_i$.
One may notice that the inner loop in algorithm \ref{alg:MRConv} can be parallelized, leading to an execution time in $\mathcal O(n)$.

We do not present here other methods, though we already mentioned Huang in \cite{Huang83}, whose approach runs in $\mathcal O(\log(n))$ time, with the same complexity. Another method could be the one presented by Wang \emph{et al.} in \cite{wang1999new}, with the same complexity properties.

\paragraph{Conversion to the redundant modulus}

We now need to know the residue $x_{m_a} \gets X \bmod m_a$. This is Algorithm \ref{alg:m_aConv}.
In this algorithm, we precompute the following values, for $i=2$ to $n$ :

$$\beta_i \gets \left|\prod_{j=1}^{i-1} m_i\right|_{m_a}$$

\begin{algorithm}
      \caption{From Mixed Radix to the residue $\bmod m_a$, $to_{m_a}(\{a_1,\cdots,a_n\})$}
  \label{alg:m_aConv}
  \begin{algorithmic}[1]
    \REQUIRE  $\{a_1,\hdots,a_n\}$  the mixed radix representation of $X$, $0\leq X<M$, precomputed values $\beta_i$
    \ENSURE $x_{m_a} = X \bmod m_a$
    \STATE $x_{m_a} \gets a_1 \bmod m_a$
    \FOR{$i$ from $2$ to $n$}
        \STATE $x_{m_a} \gets (x_{m_a} + a_i\cdot \beta_i) \bmod m_a$
    \ENDFOR
    \RETURN $x_{m_a}$
   
  \end{algorithmic}

\end{algorithm}

The instruction count is as follows: $n$ word size modular multiplications and $n-1$ word size modular additions. The memory cost is the storage of the $(n-1)$ $\beta_i$ constants, for $2\leq i\leq n$.

Again, parallelization techniques allow an execution time in $\mathcal O(\log(n))$.

%\paragraph{Comparison algorithm}

\subsection{Complexity}

%We need an extra modulus. In case of modular computations, this modulus can be a wordize $w$ modulus, which renders the probability that $X-Y \equiv 0 \bmod m_a $ unlikely, i.e. $\mathcal P(X-Y \equiv 0 \bmod m_a) \approx \frac 1{2^w}$

We provide Table \ref{tab:comp2} the operation count and the time complexities, versus a sequential or parallel implementation. We compare our work with the the classic approach, requiring two mixed-radix conversions (see Flores \cite{szabo1969}), which compares two numbers in the same conditions as ours: no special form base or moduli, no bounds on the numbers.

The memory cost is $\frac{n(n-1)}{2}$ while in this work it is $\frac{(n+2)(n-1)}{2}$.

\begin{table}[htbp]
    \centering
    \begin{tabular}{|c|c|c|c|}
        \hline
         &  \multirow{2}{1.2cm}{Op. Count}   &\multicolumn{2}{c|}{Time Comp.} \\
         \cline{3-4}
         &&seq. &   par.\\
         \hline
         \hline
         \multirow{2}{6mm}{\bf This work}&  $(\frac{n\cdot(n-1)}{2} +n) \mathcal M$&\multirow{2}{7mm}{$\mathcal O(n^2)$}&\multirow{2}{1.5cm}{$\mathcal O(\log(n))$}\\
         &$+ 2n \mathcal A$&&\\
         \hline
         \hline
         \multirow{2}{14mm}{Classic Appr.\cite{szabo1969}}&  $({n\cdot(n-1)}) \mathcal M$&\multirow{2}{7mm}{$\mathcal O(n^2)$}&\multirow{2}{1.5cm}{$\mathcal O(\log(n))$}\\
         &$+ n \mathcal C$&&\\
         \hline
    \end{tabular}
    \medskip
    \caption{Complexity of the comparison algorithm, $\mathcal M$ = word size modular multiplication, $\mathcal A$ = word size addition, $\mathcal C$ = word size comparison}
    \label{tab:comp2}
\end{table}

A comparison with other approaches is non-trivial since either the proposed methods are specific for special form of RNS base or requiring conditions on the input data (bounds, etc.) Furthermore, previous works often present complexities in the background of hardware implementation, which is not our target.

%%%%%%%%%%%%%%%%%%%%%%%%%%%%%%%%%%%%%%%%%%%
%\section{Software implementations}

%%%%%%%%%%%%%%%%%%%%%%%%%%%%%%%%%%%%%%%%%%%
\section{Conclusion}

In this work, we have proposed a method for number comparison in RNS, in the most generic context possible, requiring one redundant modulus which is often readily available in most contexts. In this method, the RNS base can be chosen arbitrarily, without conditions on the moduli form or moduli number. In addition, there are no requirements on the input numbers that can be taken in the whole RNS dynamic range. This method achieves the comparison with one single Mixed Radix conversion, with $\mathcal O(n^2)$ and can run in $\mathcal O(\log(n))$ time, in case of parallel implementation, with $\mathcal O(n^2)$ word constant memorization. To the best of our knowledge, this is competitive in our context.

\paragraph*{Future Work}

We intend to extend this work to improve division and scaling algorithm in the context of software implementations taking advantage of SIMD instruction sets, for application in various use cases (homomorphic encryption, post-quantum schemes...). For all these approaches, software implementations will attempt to take advantage of the SIMD feature and optimization with various ranges of parameters, in particular moduli number and size.

\bibliographystyle{plain}
\bibliography{biblio.bib}

\end{document}